%% file: AModelOfLayeredArchitectures.tex
\documentclass[american,submission,copyright,creativecommons,unicode=true,hidelinks=true]{eptcs}
\providecommand{\event}{FESCA 2015}
\usepackage[T1]{fontenc}
\usepackage[american]{babel}
\usepackage{amsmath}
\usepackage{amssymb}
\usepackage{amsthm}
\usepackage{stmaryrd}
\usepackage{xargs}
\usepackage{breakurl}
\usepackage{listings}
\usepackage{caption}
\usepackage{subcaption}
\usepackage{paralist}
\usepackage{colonequals}

\usepackage{tikz}
\usetikzlibrary{shapes,decorations.pathreplacing,calc}
\newcounter{defc}
\newcommand{\defc}{\refstepcounter{defc}\tag{D \thedefc}}

\include{Macros}
\title{A Model of Layered Architectures}

\author{Diego Marmsoler \qquad Alexander Malkis \qquad Jonas Eckhardt\institute{Technische Universit\"at M\"unchen\\Germany}}
\def\titlerunning{A Model of Layered Architectures}
\def\authorrunning{D. Marmsoler, A. Malkis, J. Eckhardt}
\begin{document}
\maketitle

\input{Abstract}
\input{Introduction}
\input{Model}
\input{Conclusion}
%\newpage
%\bibliographystyle{plain}
%\bibliographystyle{abbrv}
\bibliographystyle{eptcs}
\bibliography{Research}

%\newpage
%\appendix

%\section{Appendix: Conventions}
%\begin{itemize}
%\item Given a tuple $t$, we shall usually provide a name $n$ for each
%component and we write $t^{n}$ to denote component $n$ of tuple
%$t$.
%\item For sets $A$ and $B$, we denote the set of all functions from $A$ to $B$ by $A\to B$ and the set of all partial functions from $A$ to $B$ by $A\dashrightarrow B$.
%\item With $\domain f \subseteq A$ we denote the domain of a partial function $f:A \dashrightarrow B$.
%\item With $[p_{0},p_{1},\dots,p_{n}\mapsto S_{0},S_{1},\dots,S_{n}]$ we
%denote a valuation of ports $p_{0},p_{1},\dots,p_{n}$ with services
%$S_{0},S_{1},\dots,S_{n}$, respectively. Formally, $[p_{0},p_{1},\dots,p_{n}\mapsto S_{0},S_{1},\dots,S_{n}]=\lambda p\in\{p_{i}\mid i\in\mathbb{N}\land i\le n\}:\,\begin{cases}
%S_{0} & \textrm{if }p=p_{0}\\
%S_{1} & \textrm{if }p=p_{1}\\
%\vdots\\
%S_{n} & \textrm{if }p=p_{n}
%\end{cases}$
%
%\item For a set $S$ and a relation $R\subseteq S\times S$, with $\mathit{wf}(S,R)$,
%we requires $R$ to be well-founded on $S$.
%
%\item For a function $f:\, A\rightarrow B$ and a set $A'\subseteq A$,
%we write $f\restriction A':\, A'\rightarrow B$, to denote the restriction
%of function $f$ to domain $A'$.
%
%\item Sets which are assumed to exist and represent parameters of our theory are denoted by uppercase $\mathtt{MONOSPACE}$.
%\item Sets containing all elements of something are denoted calligraphic $\mathcal{S}$.
%\end{itemize}

\end{document}

%% file: Macros.tex
\newcommand{\NNP}{\mathbb{N}^+}
\newcommand{\NNZ}{\mathbb{N}_0}
\newcommand{\ZZ}{\mathbb{Z}}
\global\long\def\pset#1{\wp\left(#1\right)}

\global\long\def\val#1{\overline{#1}}

\global\long\def\lin#1{#1.\mathit{in}}
\global\long\def\lout#1{#1.\mathit{out}}

\global\long\def\proj#1#2{\sigma_{#1}\left(#2\right)}
\global\long\def\lfun#1{#1.\mathit{fun}}

\global\long\def\domain#1{\mathrm{dom}\left(#1\right)}

\global\long\def\layers#1{#1.\mathit{l}}
\global\long\def\conf#1{#1.\mathit{conf}}
\global\long\def\fconf#1#2{#1.\mathit{conf}\left(#2\right)}

\newcommandx\sel[2][usedefault, addprefix=\global, 1=]{\Pi_{#1}\left(#2\right)}
\newcommandx\syndep[3][usedefault, addprefix=\global, 1=, 2=]{#1\prec_{#3}#2}
\newcommandx\nsyndep[3][usedefault, addprefix=\global, 1=, 2=]{#1\not\prec_{#3}#2}
\newcommandx\syndept[3][usedefault, addprefix=\global, 1=, 2=]{#1\prec_{#3}^{+}#2}
\newcommandx\syndeprt[3][usedefault, addprefix=\global, 1=, 2=]{#1\prec_{#3}^{*}#2}
\newcommandx\nsyndeprt[3][usedefault, addprefix=\global, 1=, 2=]{#1\not\prec_{#3}^{*}#2}

\global\long\def\sem#1#2{\llbracket#2\rrbracket_{#1}}

\global\long\def\lupdate#1#2{\left[#1\mapsto#2\right]}
\global\long\def\update#1#2#3{#1\left[#2\mapsto#3\right]}

\newcommandx\semdep[3][usedefault, addprefix=\global, 1=, 2=]{#1\ll_{#3}#2}
\newcommandx\notsemdep[3][usedefault, addprefix=\global, 1=, 2=]{#1\not\!\ll_{#3}#2}

\newcommand{\menquote}[1]{\ensuremath{\text{\textquotedblleft} #1 \text{\textquotedblright}}}
\DeclareMathOperator{\dom}{dom}

\theoremstyle{definition}
\newtheorem{defn}{Definition}[section]
\newtheorem{exmp}[defn]{Example}
\theoremstyle{plain}
\newtheorem{lem}[defn]{Lemma}
\newtheorem{prop}[defn]{Proposition}
\newtheorem{thm}[defn]{Theorem}
\newtheorem{corl}[defn]{Corollary}

%% file: Abstract.tex
\begin{abstract}
Architectural styles and patterns play an important role in software 
engineering. 
One of the most known ones is the layered architecture style. 
However, this style is usually only stated informally, which 
may cause problems such as ambiguity, wrong conclusions, and difficulty when
checking the conformance of a system to the style.
We address these problems by providing a formal, denotational semantics of the
layered architecture style. Mainly, we present a sufficiently abstract and rigorous description
of layered architectures. 
Loosely speaking, a layered architecture consists of a 
hierarchy of layers, in which services communicate via ports. A layer is 
modeled as a relation between used and provided services, and layer 
composition is defined by means of relational composition. Furthermore, we 
provide a formal definition for the notions of syntactic and semantic 
dependency between the layers. We show that these dependencies are 
not comparable in general. Moreover, we identify sufficient conditions under which,
in an intuitive sense which we make precise in our treatment,
the semantic dependency implies, is implied by, or even coincides with
the reflexive-transitive closure of the syntactic dependency.
Our results provide a technology-independent 
characterization of the layered architecture style, 
which may be used by software architects to ensure that a system is indeed 
built according to that style. %This work aims to contribute towards a rigorous theory of architectural styles and patterns. 
%With this work we address these problems by providing a formal% treatment of one such style: the layered architecture or, in other words, the virtual machines style. %We follow an approach based on denotational semantics: first,
%, denotational semantics of one such style: the layered architecture or, in other words, the virtual machines style. First,
% we provide a general description of layered architectures. Roughly speaking,
%Finally, we identify a basic version of the style, provide a formal definition thereof in terms of constraints over our model and discuss possible implications. 
%Thus, with this work we aim to contribute towards a rigorous theory of architectural styles and patterns.
%\looseness=-1
\end{abstract}

%% file: Introduction.tex
\section{Introduction}

%With this work, we address recent calls for theory by the Software Engineering Community~\cite{Broy2011,Johnson2012}, specifically, we address a point raised by Shaw and Clements who reflect on the software architecture discipline and identify \textquotedblleft expanding formal relationships between architectural design decisions and quality attributes\textquotedblright ~\cite{Shaw2006} as one of the major future directions to go for the field.
%% AM: a splice comma is an error except in poetry. Alternatives: "; specifically, we address ..." or "specifically, addressing ...".
%% AM: More consise:
%% This work targets recent calls for theory by the Software
%% Engineering Community~\cite{Broy2011,Johnson2012}, specifically, addressing
%% a point raised by Shaw and Clements who reflect on the
%% software architecture discipline and identify \textquotedblleft expanding
%% formal relationships between architectural design decisions and quality
%% attributes\textquotedblright ~\cite{Shaw2006} as one of the major
%% future directions to go for the field.

Lack of discipline is a substantial technical source of failures in a number of software product lines \cite{Broy2011,Johnson2012} (while other sources as, e.g., bad management, also exist). A poor architecture 
can result in a disaster for the whole project \cite{Garlan-SoftwareArchitectureARoadmap}, hence, \textquotedblleft expanding 
formal relationships between architectural design decisions and quality 
attributes\textquotedblright\ \cite{Shaw2006} has been identified as a 
promising future direction to go for the field. 
We address the lack of discipline in architectural design 
\cite{BassClementsKazman-SoftwareArchitectureInPractice} 
by providing a formal model for one of the most important architectural styles, 
namely, the layered architecture style, 
which is also known as the virtual machines style.

While this work  contributes to a rigorous theory of architecture styles,
we believe that it has also implications for the practicing architecture
researcher and the prospective software architect. The software architecture
researcher can rely on a mathematical model when working with
styles, while the prospective architect is provided with a solid foundation
for her/his work. A theory of styles would provide the architect with
a set of properties which allows her/him to decide whether a system is actually 
built according to a specific style, in our case, the layered architecture
style. Moreover, the outcome of the analysis would provide the architect
with a set of properties she/he can rely on from a system built according
to a style, for example, semantic independence of lower-level layers
from upper-level layers for systems built according to the layered architecture
style.

\subsection{Approach}

In previous work~\cite{Marmsoler2014}, we describe an approach to formalize architectural
styles. Based on the insight that each style requires its own semantic domain \cite{Abowd1995}, this approach roughly follows three main steps:
\begin{itemize}
\item Find a mathematical model which reflects the nature of the style.
This is probably the most difficult part, since the model must reflect
the fundamental characteristics of a style. It should be as abstract
as possible to allow the results of later analyses to be applied to
a broad range of systems. If for some style %, 
%% AM: I think that a comma here would be wrong.
an adequate model already exists, this step can be skipped.
\item Provide a set of axioms for the model which constrain its structure.
Through the addition of new axioms it is possible to specialize a
style and investigate variations thereof. For example, in the layered
architecture style, a configuration is usually isomorphic to a directed acyclic graph. However, we could add an axiom which restricts
configuration to a directed sequence of layers to get a description
of the strict version of the style.
\item Finally, we can analyze a style by means of mathematical proofs.
We can state characteristic properties for a style and prove them
from our model.
\end{itemize}
In the following we apply our approach to the layered architecture
style as described in \cite{Buschmann2007,Shaw1996,Taylor2010}. %Thus, in Section~\ref{sec:BGRW} we first provide some background and discuss related work in this area. 

Our \emph{major contributions} are:
\begin{compactitem}
\item an abstract and nonetheless precise notion of a layer (for this moment, one can loosely think of a layer as a provider of services that uses some other services),
\item a notion of a layered architecture configuration, which is a collection of layers connected via ports (detailed later in the paper),
\item a denotational semantics of a layered architecture configuration,
\item a model for updating a layer, i.e., changing its semantics,
\item for a pair of layers of a layered architecture configuration, the notions of
\begin{compactitem}
\item a syntactic dependency and
\item a semantic dependency,
\end{compactitem}
\item examples on which the dependencies differ,
\item the following (for now intuitively stated) link between the dependencies:
\begin{compactitem}
\item in any layered architecture configuration the semantic dependency implies the reflexive-transitive closure of the syntactic dependency,
\item in any so-called usable layered architecture configuration the semantic dependency is equivalent to the reflexive-transitive closure of the syntactic dependency.
\end{compactitem}
\end{compactitem}
%Then we provide a model for layered architectures in Section~\ref{sec:model}: first we define the basic notion of layer and configuration, and provide a characterization of layer composition. Then we define the notion of syntactic and semantic dependency and investigate their relationships. Finally, we conclude our analysis in Section~\ref{sec:conclusion} with a brief discussion of our results and possible directions for future work.

\section{Background and Related Work}\label{sec:BGRW}

Related work can roughly be categorized in three main areas: approaches to formalization of architectural styles,
informal descriptions of the layered architecture style, and existing formal analyzes of architectural styles.

In analyzing architectural styles, our work is actually based on work
regarding \emph{approaches to formalization of architectural styles}.
In our work, we follow an approach based on Abowd et al.~\cite{Abowd1995}.
In that work, the authors apply the general approach of denotational
semantics to software architectures with the fundamental insight that
each architectural style needs its own semantic model. On this basis,
Allen~\cite{Allen1997} provides an architecture description language
based on CSP~\cite{Hoare1985} to allow the specification and analysis
of architectural styles. A different, though related approach is
provided by Moriconi et al.\ in~\cite{Moriconi1995}. There, the authors
use first order logical theories to describe architectural styles
and they suggest to use the concept of faithful interpretation mappings
to relate different styles. In a third approach, Le~Métayer in~\cite{LeMetayer1998} proposes to describe
architectures as graphs and architectural styles as graph grammars
with the aid of analyzing architecture evolution.
%An interesting approach applying category theory to formalize architectural concepts is provided by Fiadeiro in~\cite{Fiadeiro2005}.
Finally, Bernardo et al.~\cite{Bernardo2000} propose the use of process algebras to
formalize architectural types, which are weaker forms of architectural styles. \looseness=-1

To build our model for layered architectures, we heavily rely on the
intuition provided by \emph{informal descriptions of the layered architecture
style}. Some of the first documented descriptions of the style can
be found in the work of Shaw and Garlan~\cite{Shaw1996}, where they
identify a set of %well known
%% AM: write with a hyphen before a noun.
well-known styles observed in industry. Taylor et al.~\cite{Taylor2010} elaborate on that work and distinguish between two kinds of layered architectures: the virtual machines style and
the client-server style. Finally, there exists much literature from
practicing architects documenting architectural styles and patterns.
%In our work, we considered 
We consider~\cite{Buschmann2007} as one well-known
representative %for
%% AM: ``representative for'' sounds strange.
of this kind of works.

%% AM: Avoid repeating ``work'' too often.
%While all these works provide the necessary background for our work, there is other work
While all such references provide the necessary background for our study, there is another line of research 
 on \emph{existing formal analyzes of architectural styles}
which is closely related to our work.
In~\cite{Garlan1991}, Garlan and Notkin provide a formal basis for the
implicit-invocation architectural style. The signal-processing style is
analyzed by Garlan and Norman in~\cite{Garlan1990}. Moreover, we can find a
formal description of the pipes-and-filters style in the work of Allen and
Garlan~\cite{Allen1992} and in Broy's $\mathtt{Focus}$-theory~\cite{Broy2001}.
The Enterprise Java Beans architectural style is formally analyzed
by Sousa and Garlan in~\cite{Sousa2001}. In~\cite{Moriconi1995},
the data-flow style is related to the pipes-and-filters style, the
batch-sequential style, and the shared-memory style. The client-server
style is described by Le~Métayer in~\cite{LeMetayer1998}. Finally, there are some formal analyzes of the layered architecture style.
In~\cite{Zave13}, Zave and Rexford build a formal model of layered architectures and use the Alloy Analyzer~\cite{Jackson2012} to analyze the style. Since their analysis concentrates on network-specific properties, it is a refinement of our model, thus, complementing our work.
The work which is probably closest to our work is the one of Broy which provides
a better understanding of the layered architecture style in~\cite{Broy2005}.

In \cite{Broy2005}, Broy provides a model of services and of layered
architectures based on the $\mathtt{Focus}$ theory~\cite{Broy2001}.
In that model, a layer is a component with an import and an export
interface and a layered architecture is a stack of several layers.
Although that model is an important contribution towards a better
understanding of layered architectures, the model 
%is actually based on a concrete model of computation based on streams. 
represents computations explicitly using streams. 
Our model abstracts further away from such details of computations, concentrating  
%However, we think that a model for an architectural style should abstract from the details of the %computational model 
%computation 
% and concentrate 
on the major characteristics
of the style, thus making the results applicable to several, different
%models.
representations of computations.
% This is why 
In fact, our model is based on an abstract notion of a service,
and streams are just one possible realization thereof as shown in Ex.~\ref{ex:stateful}. Other realizations include stateless services as shown in Ex.~\ref{ex:stateless} and more complex interactions as shown in Ex.~\ref{ex:complex}.

%% file: Model.tex
\section{A Model for Layered Architectures}\label{sec:model}

In the following section we provide a model of layered architectures
based on ports and services. With this model we want to provide the
basis for a rigorous analysis of the style. Therefore, the model should
be as abstract as possible and capture the intuitive understanding
of the style to allow formulation of characteristic properties of
the style.

%Ports can provide or require services and are classified into input
%and output, respectively. A layer is then modelled as a function from
%valuations of its input-ports to a set of valuations of its output
%ports. Syntactic equivalence of layers is discussed. A layered architecture
%configuration is then modelled as a family of layers, so called layer
%instances, and an attachment relation over input and output-port instances.
%Then, the notion of syntactic dependency, layer semantics, semantic
%change and semantic dependency is discussed. Moreover, a characterisation
%of layer composition is provided. 

\subsection{Ports and Services}

For our model of layered architectures, we assume the existence of
sets $\mathtt{PORT}$ and $\mathtt{SERVICE}$ which contain all ports and services, respectively. Thereby, our notion of service
is rather abstract; A service can be anything, from a simple method
to a complex web-service consisting of a series of interactions. A
port is a placeholder for a set of related services; one can think of
the method's signature or of the address of the web service. Thus, we assume
the existence of a function \ $\mathit{type}\colon\mathtt{PORT}\rightarrow\pset{\mathtt{SERVICE}}$, \ (where $\pset{X}$ is the power set of a set $X$) which assigns a type to each port.
That is, the type of a port is simply a set of services. 
We require that each port is classified either as an input port or as an output port, but not as both. We let $\mathcal{I}$ be the set of input ports and $\mathcal{O}$ be the set of output ports:
\[
\mathcal{I}\cup\mathcal{O} \ = \ \mathtt{PORT}\quad \text{and} \quad
\mathcal{I}\cap\mathcal{O} \ = \ \emptyset\,.
\]
Ports and services constitute the parameters of our theory. By saying what port and services are, our theory can be applied to different contexts.

In the following, let $\NNP$ be the set of positive integers, $\ZZ$ the set of all integers. 
\begin{exmp}[A model of stateless services]\label{ex:stateless}
\lstset{language=C++, columns=fullflexible}%
\begin{figure}%
\begin{subfigure}[t]{.3\textwidth}\centering%
\begin{tikzpicture}[iport/.style={circle,fill,inner sep=3pt},oport/.style={circle,draw,fill=white,inner sep=3pt},
layer/.style={draw,minimum height=30pt,minimum width=100pt}]

\node[layer] (l1) at (0,0) {%
\begin{lstlisting}
bint mult(bint x, bint y) {
  bint z := 0;
  while (y>0) {
    z := add(x,z);
    y := sub(y,1);
  }
  return z;
}
\end{lstlisting}};
\node [oport, outer sep=2pt] (op) at (l1.east) {};
\node at ($(op)+(7pt,7pt)$) {\Large$o$};
\node [rectangle, rounded corners=9, outer sep=0, draw, minimum width=106pt, minimum height=16pt] (os) at ($(l1.north west)+(1.92,-.51)$) {};
\draw (os.south east) -- (op);

\node [iport, outer sep=2pt] (ip1) at ($(l1.west)+(0,25pt)$) {};
\node[rectangle, rounded corners=6, outer sep=0, draw, minimum width=43pt, minimum height=12pt] (is1) at ($(l1.west)+(2.21,.18)$) {};
\node at ($(ip1)+(-.3,.3)$) {\Large$i_1$};
\draw (is1.north west) -- (ip1);

\node [iport, outer sep=2pt] (ip2) at ($(l1.west)-(0,25pt)$) {};
\node[rectangle, rounded corners=6, outer sep=0, draw, minimum width=43pt, minimum height=12pt] (is2) at ($(l1.west)+(2.21,-.31)$) {};
\node at ($(ip2)+(-.2,-.2)$) {\Large$i_2$};
\draw (is2.south west) -- (ip2);
\end{tikzpicture}%
\caption{Stateless}\label{subfig:simple_stateless}%
\end{subfigure}
\qquad
\begin{subfigure}[t]{.6\textwidth}\centering
\begin{tikzpicture}[iport/.style={circle,fill,inner sep=3pt},oport/.style={circle,draw,fill=white,inner sep=3pt},
layer/.style={draw,minimum height=50pt,minimum width=160pt}]

\node[layer] (l1) at (0,0) {%
\begin{lstlisting}
static map<pair<bint,bint>,bint> cache; // lookup table
bint mult(bint x, bint y) {
  pair<bint,bint> p := make_pair(x,y);
  if(cache.count(p)>0) //  if lookup successful
    return cache[p];
  else {  // if the current input has not been cached yet
    bint z := 0;
    while (y>0) {
      z := add(x,z);
      y := sub(y,1);
    }
    cache[p] := z; // store the result
    return z;
  }
}
\end{lstlisting}};
\node [oport, outer sep=2pt] (op) at (l1.east) {};
\node at ($(op)+(7pt,7pt)$) {\Large$o$};
\node [rectangle, rounded corners=9, outer sep=0, draw, minimum width=107pt, minimum height=16pt] (os) at ($(l1.north west)+(1.92,-1)$) {};
\draw (os.south east) -- (op);

\node [iport, outer sep=2pt] (ip1) at ($(l1.west)$) {};
\node[rectangle, rounded corners=6, outer sep=0, draw, minimum width=43pt, minimum height=12pt] (is1) at ($(l1.west)+(2.56,-.54)$) {};
\node at ($(ip1)+(-.3,.3)$) {\Large$i_1$};
\draw (is1.north west) -- (ip1);

\node [iport, outer sep=2pt] (ip2) at ($(l1.west)-(0,2)$) {};
\node[rectangle, rounded corners=6, outer sep=0, draw, minimum width=43pt, minimum height=12pt] (is2) at ($(l1.west)+(2.56,-1.04)$) {};
\node at ($(ip2)+(-.2,-.2)$) {\Large$i_2$};
\draw (is2.south west) -- (ip2);
\end{tikzpicture}%
\caption{Stateful}\label{subfig:simple_stateful}%
\end{subfigure}
\caption{Stateful and stateless layers. The programming language type \lstinline!bint! represents large integers.}
\end{figure}
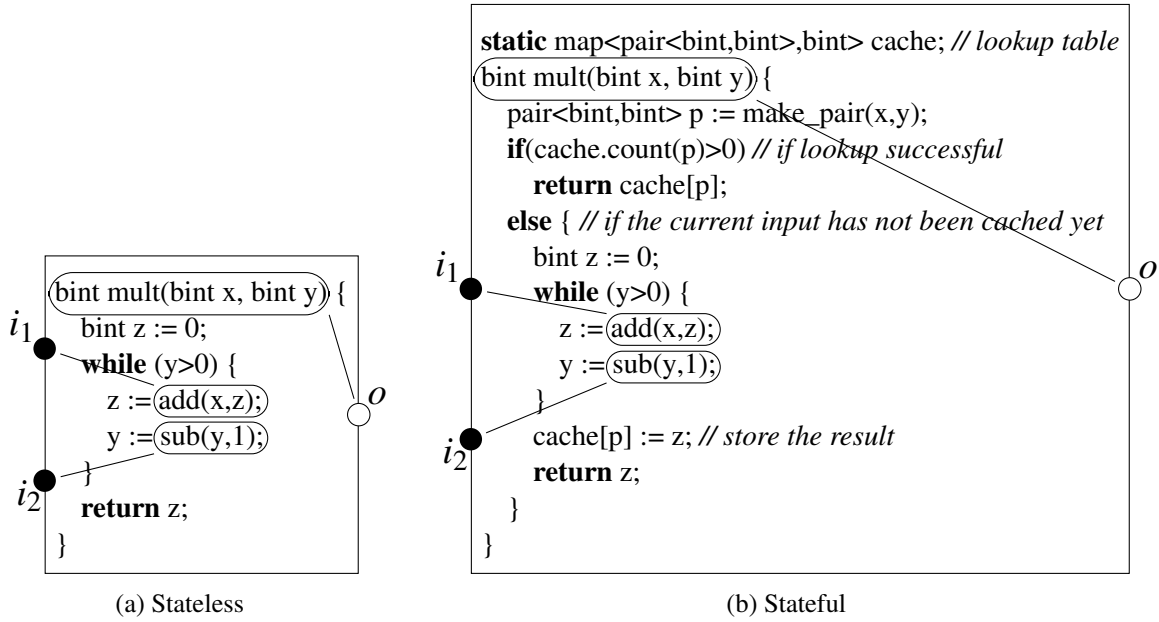
  Consider the code depicted in Fig.~\ref{subfig:simple_stateless}, where we write \lstinline!bint! for \lstinline!bigint!, a programming language type of large but fixed-size integers. In this example, we define the set of input ports as $\mathcal{I}=\{i_1,i_2\}$, the set of output ports as $\mathcal{O}=\{o\}$, where the ports are signatures, i.e., simplifying, strings:\\[.5ex]
$i_1 \ = \ \menquote{\mathtt{bint}\ \mathtt{add}(\mathtt{bint},\mathtt{bint})}\,,\quad
 i_2 \ = \ \menquote{\mathtt{bint}\ \mathtt{sub}(\mathtt{bint},\mathtt{bint})}\,,\quad 
 o \ = \ \menquote{\mathtt{bint}\ \mathtt{mult}(\mathtt{bint},\mathtt{bint})}\,.$\\[.5ex]
Here, a service at a port will be a (set-theoretic) map whose signature is given by the port. For the sake of the example, let us fix some $M\in\NNP$ and use $\mathit{bint}=[-2^M,2^M-1]$ for the set of representatives of integers modulo $2^{M+1}$, writing $\widehat{a}$ for the representative of $a\in\ZZ$. 
%% We define:\\[.5ex]
%% $\begin{array}{r@{\ }c@{\ }l}
%% \mathit{type}(i_1) & = & \{\text{the unique map in}\ (\mathit{bint}{\times}\mathit{bint}\to\mathit{bint})\ \text{representing addition in modular arithmetics}\}\,,\\
%% \mathit{type}(i_2) & = & \{s\in(\mathit{bint}{\times}\mathit{bint}\dashrightarrow\mathit{bint})\mid \forall\,y\in\ZZ\colon \widehat{y}>0\ \Rightarrow\ ((\widehat{y},\widehat{1}),\widehat{y-1})\in s\}\,,\\
%% \mathit{type}(o) & = & \{m\in(\mathit{bint}{\times}\mathit{bint}\to\mathit{bint})\mid \forall\,x,y\in\ZZ\colon \widehat{y}\ge 0\ \Rightarrow\ m(\widehat{x},\widehat{y})=\widehat{xy}\}\,.
%% \end{array}$\\[.5ex]
%% The maps in the types of $i_1$  and $o$ are total, meaning that the computations at the input port $i_1$ are expected to terminate for all inputs and that the provided computation is guaranteed to terminate for all inputs. Our types say that
%% \begin{compactitem}
%% \item at port $i_1$, exactly the modular addition is expected,  which terminates for all inputs,
%% \item at port $i_2$, a service is expected that returns the difference of two number representations under certain conditions and we know nothing if the conditions are not met,
%% \item at port $o$, services are provided that always terminate and each of the services computes multiplication for certain inputs. 
%% \end{compactitem}
The types of the ports are sets containing certain partial or total maps from $\mathit{bint}{\times}\mathit{bint}$ to $\mathit{bint}$:
\begin{compactitem}
\item The type of $i_1$ is the singleton set containing exactly the modular addition, which is defined for all arguments.
\item The type of $i_2$ is the set containing partial and total maps $s$ whose result coincides with that of modular subtraction whenever their first argument is positive and the second is 1.
\item The type of $o$ is the set of total maps $m$ that multiply the arguments $x$, $y$ modulo $2^{M+1}$ whenever $y$ is nonnegative. Such $m$ must return some result also for negative $y$.
\end{compactitem}
In this example, the types abstract away some details about termination and outcome and all details about the way the computations are performed.\qed
\end{exmp}
The above example does not use any global state. In a more complex model, a layer may also have an encapsulated state, as it is the case for object-oriented programming languages. We can easily encode stateful models by changing the notion of service to relate \emph{streams} of concrete values of the input parameters to \emph{streams} of concrete return values, as we will see in Ex.\ \ref{ex:stateful}.

In the following, let $\NNZ$ be the set of nonnegative integers and $\dom f$ the domain of a (partial) map $f$.
\begin{exmp}[A model of stateful services]\label{ex:stateful}
The code from Ex.~\ref{ex:stateless} is slow. For the sake of the example, let us assume that some calls to \lstinline!mult! are often repeated with the same arguments so that caching would help reducing the running time, and let us cache every input-output pair in a simple way as in Fig.~\ref{subfig:simple_stateful}. In the worst case the cache grows until the memory is exhausted, after which we assume that cache insertion and all later events may block or have an arbitrary behavior.

We assume that the cache operations are purely internal and that the cache can hold at least $N$ input-output pairs. We define the set of input ports as $\mathcal{I}=\{i_1,i_2\}$ and the set of output ports as $\mathcal{O}=\{o\}$ again. Let $\mathit{sbint}$ $=$ $(\mathrm{stream}\ \mathit{bint})$ $=$ $\mathit{bint}^*\cup\mathit{bint}^{\omega}$ be the set of streams over $\mathit{bint}$, i.e., the set of finite and countably infinite sequences over $\mathit{bint}$, where we index the elements of a stream by the corresponding downward-closed subset of $\NNZ$. We lift the previous types of $i_1$ and $i_2$ pointwise to streams as usual; e.g., $\mathit{type}(i_1)=\{a\in(\mathit{sbint}{\times}\mathit{sbint}\to\mathit{sbint})\mid \forall\,r,s,t\in\mathit{sbint}\colon a(r,s)=t\Rightarrow (\dom t=(\dom r)\cap(\dom s)\land \forall\,i\in\dom t\colon t(i)=\widehat{r(i){+}s(i)})\}$. We define $\mathit{type}(o)$ to be the set of all maps $m\in(\mathit{sbint}{\times}\mathit{sbint} \to \mathit{sbint})$ such that whenever $m(r,s)=t$ for streams $r,s,t\in\mathit{sbint}$ and $i\in(\dom r)\cap(\dom s)$ is such that the number of cached entries $\lvert\{(r(j),s(j))\in\mathit{bint}^2\mid j{\le}i\land j\in(\dom r)\cap(\dom s)\}\rvert$ is below $N$, then $i\in\dom t$ and we have $(\widehat{s(i)}\ge 0\ \Rightarrow\ t(i)=\widehat{r(i)s(i)})$.%, we have $\dom t\subseteq (\dom r)\cap(\dom s)$, and moreover if $i\in(\dom r)\cap(\dom s)$ is such that the number of cached entries $\lvert\{(r(j),s(j))\in\mathit{bint}^2\mid j{\le}i\land j\in(\dom r)\cap(\dom s)\}\rvert$ is below $N$, then the result $t(i)$ is defined and we have $(\widehat{s(i)}\ge 0\ \Rightarrow\ t(i)=\widehat{r(i)s(i)}).$

Loosely speaking, types containing functions over streams have just helped specifying stateful abstractions of stateful services without actually referring to their state spaces.
\qed
\end{exmp}

\begin{exmp}[A model of complex services]\label{ex:complex}
	In Ex.~\ref{ex:stateful}, a service is still realized by a simple method which depends on a global state. However, we could also think about models in which a service is actually realized by a series of method calls, coordinated by some kind of protocol. By adjusting the concrete notion of service, our theory can also be applied to those kind of models. Here, the behavior of the services relates streams of concrete values for all the input parameters of \emph{all} the methods in the series with streams of output values of \emph{all} the return values of the series. \emph{Ports} are then a set of method signatures equipped with an expected order of execution. Again, an \emph{output port} specifies methods which can be called within the layers implementation while an \emph{input port} specifies those methods realized by a layer.\qed
\end{exmp}

\subsection{Valuations\label{sub:Valuations}}

For a set of ports $P\subseteq\mathtt{PORT}$, a valuation is a function
from the set $P$ to the set of services that respects the types of the ports. 
By $\overline{P}$ we denote the set of all valuations for $P$, formally,
\[
\val P=\prod_{p\in P}\mathit{type}(p)\,.
%\left(P\rightarrow\mathtt{SERVICE}\right),\defc
\]
Sometimes, we shall use $[p_{0},\dots,p_{n}\mapsto S_{0},\dots,S_{n}]$ to
denote a valuation of ports $p_{0},\dots,p_{n}$ with services $S_{0},\dots,S_{n}$, respectively. Formally, 
\[[p_{0},\dots,p_{n}\mapsto S_{0},\dots,S_{n}]\quad = \quad \lambda p\in\{p_{i}\mid i\in\NNZ\land i\le n\}.\begin{cases}
 S_{0} & \textrm{if }\ p=p_{0},\vspace{-7pt}\\
 \vspace{-5pt}\vdots\\
 S_{n} & \textrm{if}\ p=p_{n}.
\end{cases}\]

\begin{exmp}[Valuations for services]
	In the models described in Ex.~\ref{ex:stateless},~\ref{ex:stateful}, and~\ref{ex:complex}, a \emph{port valuation} just associates a services behavior with the corresponding method signature.\qed
\end{exmp}

%Add if space is left
%An aside on connection to program analysis should be made. For a moment, let us equip $\pset{\val{P}}$ with the subset order and $\valND{P}$ with the pointwise order over the subset order (i.e., for $\nu,\mu\in\valND{P}$, let $\nu\leq\mu$ iff $\forall p\in P\colon \nu(p)\subseteq\mu(p)$). The map $\alpha_f:\pset{\overline{P}} \to \valND P$, $X\mapsto\lambda p.\{\mu(p)\mid \mu{\in}X\}$, called the  \emph{dependence-free abstraction for functions}, can be shown to be the lower adjoint of a monotone Galois connection between complete lattices \cite{CousotCousot-FormalLanguageGrammarAndSetConstraintBasedProgramAnalysisByAbstractInterpretation%,Malkis-TechReportRefinementWithExceptions
%%PUT ourselves in after acceptance.
%}. %When $\overline{P}$ and $\valND P$ are viewed as Cartesian products, $\alpha_f$ is called \emph{Cartesian abstraction}. Historically, it arose from replacing the \emph{relational method} of analyzing programs by the \emph{independent attribute method} \cite{JonesMuchnik_ComplexityOfFlowAnalysisInductiveAssertionSynthesisAndALanguageDueToDijkstra}.

\subsection{Layers}
%A layer consists of input and output-ports and a behavior which connects valuations of output-ports and valuations of input-ports. Thus, a layer is modeled as a triple consisting of input and output-ports and a behavior which is a mapping from input-port valuations to output-port valuations. 
Informally speaking, a layer consists of input ports, output ports, and some behavior that generates services at output ports from services at input ports. 
The behavior may be nondeterministic, so we represent it by a map that assigns a set of output-port valuations to every input-port valuation.

%The layers should respect types in a certain sense which we now make precise. 
\begin{defn}\label{def:layer}
%A map $f\colon\val{I}\to\pset{\val O}$ for some $I,O\subseteq\mathtt{PORT}$ is \emph{proper} iff  
%$%\label{ax:ProperMap}
%\forall\mu\in\val{I}\colon
%\left(\mu\ \text{is properly typed}\ \Rightarrow\ \forall\mu'{\in}f(\mu)\colon \mu'\ \text{is properly typed}\right).
%%\conc
%$
A \emph{layer} is a triple $\left(I,O,f\right)$, where $I\subseteq\mathcal{I}$,
$O\subseteq\mathcal{O}$, and $f\colon\val I\to \pset{\val O}$.
\end{defn}
For a layer $l=\left(I,O,f\right)$, we denote by $\lin l$ its input-ports
$I$, by $\lout l$ its output-ports $O$, and by $\lfun l$ its behavior
function $f$. We denote the set of all layers by $\mathcal{L}$.\looseness-1

%To deal with nondeterministic input-port valuations,
%we define for each layer $l$ a function $\lfunND l\colon\valND{\lin l}\rightarrow\pset{\val{\lout l}}$
%as follows:
%\thinmuskip 0mu \thickmuskip 2mu \medmuskip 2mu
%\[
%\lfunND l(\mu)=\left\lbrace \mu' \in \val {\lout l} \mid \exists\nu\in\val{\lin l}\colon\left(\forall p\in\lin l\colon\nu(p)\in\mu(p)\right)\land\mu'\in\lfun l(\nu)\right\rbrace \defc
%\]
\newcommand{\idfinexample}[0]{f(i) = \{o \in \val O \mid o(o_1) = i(i_1) \wedge o(o_2) = i(i_2)\}}
\begin{exmp}[A simple layer]
Consider, for example, the layer depicted in Fig.\ \ref{fig:Layer}, which just copies $i_j$ to $o_j$ for $j \in \{1,2\}$.
In our model, such a layer is represented as a triple $\left(I,O,f\right)$
with input-ports $I=\{i_{1},i_{2}\}$, output-ports $O=\{o_{1},o_{2}\}$
and behavior function $f\in\left(\val I\rightarrow\pset{\val O}\right)$ with $\idfinexample$.\qed
\end{exmp}
\begin{figure}\vspace{-.5ex}%
\centering%
\begin{tikzpicture}[iport/.style={circle,fill,inner sep=3pt},oport/.style={circle,draw,fill=white,inner sep=3pt}, layer/.style={draw,minimum height=20pt,minimum width=150pt}]
\node[layer] (l1) at (0,0) { $f\in\left(\overline{I}\rightarrow\pset{\overline{O}}\right) \text{ with } \idfinexample$};
\node [oport] at ($(l1)+(-45pt,10pt)$) {}; \node at ($(l1)+(-55pt,17pt)$) { $o_1$};
\node [oport] at ($(l1)+(45pt,10pt)$) {}; \node at ($(l1)+(55pt,17pt)$) { $o_2$};
\draw [decorate, decoration={brace, amplitude=5pt}] (-45pt,20pt) -- (45pt,20pt); \node at ($(l1)+(0pt,30pt)$) { $O$};
\node [iport] (l1o1) at ($(l1)-(45pt,10pt)$) {}; \node at ($(l1)-(55pt,17pt)$) { $i_1$};
\node [iport] (l1i1) at ($(l1)+(45pt,-10pt)$) {}; \node at ($(l1)+(55pt,-17pt)$) { $i_2$};
\draw [decorate, decoration={brace, amplitude=5pt}] (45pt,-20pt) -- (-45pt,-20pt); \node at ($(l1)-(0pt,30pt)$) { $I$};
\end{tikzpicture}%
\vspace{-1ex}%
\caption{\label{fig:Layer}A layer with input-ports $I$, output-ports $O$
and behavior function $f$.}%
\vspace{-0.5ex}%
\end{figure}
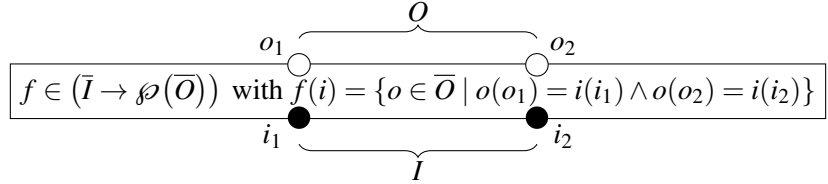

\subsection{Layered Architecture Configuration}

A layered architecture configuration consists of a set of layers and an attachment describing the connections between the layers.
%of layer input ports and output ports, respectively. 
Thus, a layered architecture configuration is modeled as a pair of a set of layers and a so-called attachment relation describing which output ports of which layers convey services to which input ports of which layers.

In the following, we denote by\vspace{-1ex}
\[X\dashrightarrow Y \quad = \quad \{f\subseteq X{\times}Y\mid \forall\, x,y_1,y_2\colon ((x,y_1)\in f\wedge(x,y_2)\in f)\Rightarrow y_1=y_2\}\vspace{-1ex}\]
the set of partial maps from a set $X$ to a set $Y$. 
\begin{defn}\label{def:lac}
A \emph{layered architecture configuration} is a pair $\left(L,A\right)$,
where $L \subseteq \mathcal{L}$ and %$A\subseteq \mathtt{PORT} \times \mathtt{PORT}$
$A\in ((\bigcup_{l\in L} \lin{l})\dashrightarrow(\bigcup_{l\in L} \lout{l}))$, called the \emph{attachment}, are such that the following constraints hold.
\begin{itemize}
\item Different layers do not share any ports, formally:
\[
\forall\ k,l\in L\colon \quad k=l \ \vee \ (\lin k \cup \lout k)\cap(\lin l \cup \lout l)=\emptyset.
\]
\item %If one layer can provide a service to another layer, the latter must be able to consume it,
If a service is provided at an output port that is connected to an input port, the layer owning the input port must be able to employ the service, 
i.e.\ the port types are compatible. Formally:
\[
\forall\ (p_i,p_o)\in A\colon \quad \mathit{type}(p_o)\subseteq\mathit{type}(p_i).
\]
\end{itemize}
\end{defn}\noindent
For a layered architecture configuration $c=(L,A)$, we denote the set of
layers $L$ by $\layers c$ and the attachment relation $A$ by $\conf c$.

The domain of the attachment is a subset of the occurring input-ports, and the range is a subset of the occurring output-ports, signifying that the input ports are connected to the output ports. The attachment is a partial map, since not necessarily all input ports are internally connected, but whenever an input port is connected, it accepts services only from one output port.

\begin{exmp}[A simple layered architecture configuration]
Fig.\ \ref{fig:layered-architecture-configuration} shows a layered
architecture configuration $c=(L,A)$. The first component of the layered architecture configuration describes the layers, i.e., their input and output-ports and their behavior function. In this example $L=\{l_0,\ldots,l_n\}$, where $l_k = \left(I_k,O_k,f_k\right)$ with $I_k = \{i_{0,k}, i_{1,k}, i_{2,k}\}$ for $0 < k\le n$ and $I_0 = \{i_{0,0}\}$, $O_k = \{o_{0,k}, o_{1,k}, o_{2,k}\}$ for $0\le k < n$ and $O_n = \{o_{0,n}\}$, and $f_k \in \left(\val {I_k}\rightarrow\pset{\val {O_k}}\right)$ for $0 \le k \le n$.

The second component of the layered architecture configuration describes the attachment relation $A$ which relates $i_{1,k}$ with $o_{1,k{+}1}$ and $i_{2,k}$ with $o_{2,k{+}1}$:
$A=\{(i_{j,k}, o_{j,k{-}1})\mid j\in\{1,2\} \land k\in\{1,\ldots,n\}\}$.
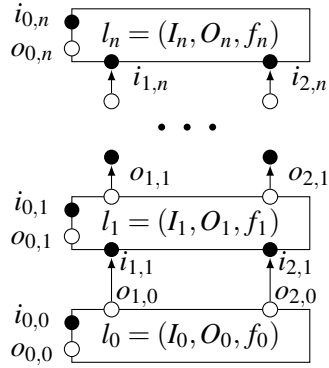
\begin{figure}[htbp]%
\centering%
\begin{tikzpicture}[iport/.style={circle,fill,inner sep=2pt},oport/.style={circle,draw,fill=white,inner sep=2pt},
layer/.style={draw,minimum height=20pt,minimum width=90pt}]
%\draw[dashed]  (-3.5,5.5) rectangle (2,0.5);
iffalse Layer Ln--------------------------------------------------------------- fi

\node[layer] (ln) at (-0.5,5) {$l_n=(I_n,O_n,f_n)$};

\node [iport] at ($(ln)+(-45pt,5pt)$) {};
\node at ($(ln)+(-60pt,7pt)$) {$i_{0,n}$};

\node [oport] at ($(ln)-(45pt,5pt)$) {};
\node at ($(ln)+(-60pt,-7pt)$) {$o_{0,n}$};

\node [iport] (lni1) at ($(ln)-(30pt,10pt)$) {};
\node at ($(ln)-(15pt,17pt)$) {$i_{1,n}$};
\node [oport] (lni1out) at ($(ln)-(30pt,25pt)$) {};
\draw[-latex] (lni1out) edge (lni1) ;

\node [iport] (lni2) at ($(ln)+(30pt,-10pt)$) {};
\node at ($(ln)+(45pt,-17pt)$) {$i_{2,n}$};
\node [oport] (lni2out) at ($(ln)+(30pt,-25pt)$) {};
\draw[-latex]  (lni2out) edge (lni2);

iffalse dots--------------------------------------------------------------- fi
\node at (-0.5,3.75) {\Huge $\cdots$};

iffalse Layer L1--------------------------------------------------------------- fi

\node[layer] (l1) at (-0.5,2.5) {$l_1=(I_1,O_1,f_1)$};

\node [iport] at ($(l1)+(-45pt,5pt)$) {};
\node at ($(l1)+(-60pt,7pt)$) {$i_{0,1}$};

\node [oport] at ($(l1)-(45pt,5pt)$) {};
\node at ($(l1)+(-60pt,-7pt)$) {$o_{0,1}$};

\node [iport] (l1o1) at ($(l1)+(-30pt,25pt)$) {};
\node [oport] (l1o1in) at ($(l1)+(-30pt,10pt)$) {};
\node at ($(l1)+(-15pt,17pt)$) {$o_{1,1}$};
\draw[-latex]  (l1o1in) edge (l1o1);

\node [iport] (l1o2) at ($(l1)+(30pt,25pt)$) {};
\node [oport] (l1o2in) at ($(l1)+(30pt,10pt)$) {};
\node at ($(l1)+(45pt,17pt)$) {$o_{2,1}$};
\draw[-latex] (l1o2in) edge (l1o2) ;

\node [iport] (l1i1) at ($(l1)-(30pt,10pt)$) {};
\node at ($(l1i1)+(2ex,-1.15ex)$) {$i_{1,1}$};

\node [iport] (l1i2) at ($(l1)+(30pt,-10pt)$) {};
%\node at ($(l1)+(45pt,-15pt)$) {$i_{2,1}$};
\node at ($(l1i2)+(2ex,-1.15ex)$) {$i_{2,1}$};

iffalse Layer L0--------------------------------------------------------------- fi

\node[layer] (l0) at (-0.5,1) {$l_0=(I_0,O_0,f_0)$};

\node [iport] at ($(l0)+(-45pt,5pt)$) {};
\node at ($(l0)+(-60pt,7pt)$) {$i_{0,0}$};

\node [oport] at ($(l0)-(45pt,5pt)$) {};
\node at ($(l0)+(-60pt,-7pt)$) {$o_{0,0}$};

\node [oport] (l0o1) at ($(l0)+(-30pt,10pt)$) {};
\node at ($(l0o1)+(2ex,1ex)$) {$o_{1,0}$};

\node [oport] (l0o2) at ($(l0)+(30pt,10pt)$) {};
\node at ($(l0o2)+(2ex,1ex)$) {$o_{2,0}$};

\draw[-latex] (l0o1) edge (l1i1) ;
\draw[-latex] (l0o2) edge (l1i2);
\end{tikzpicture}%
\caption{\label{fig:layered-architecture-configuration}A layered architecture configuration with $n+1$ layer instances.}%
\end{figure}\qed
\end{exmp}

\subsubsection{Selection and Projection}
To facilitate reasoning about layered architecture configurations, in the following we introduce two kind of operators: selection and projection operators.

A selection operator allows to access ports belonging to a layered architecture configuration.
\begin{defn}
For a layered architecture configuration $c$, we define port selection as follows:
\[
	\sel[i] c = \bigcup_{l\in\layers c}\lin l~\mathrm{and}~\sel[o] c = \bigcup_{l\in\layers c}\lout l.
\]
\end{defn}\noindent
To select all ports of a layered architecture configuration, we just write
\[
\sel c = \sel[i] c \cup \sel[o] c.\defc
\]
To select only the \emph{open input-ports} (input ports which are not attached) of a layered architecture configuration $c$, we write
\[
	\sel[\mathit{in}] c \ \colonequals \ \sel[i] c \setminus \domain{\conf c}.\defc
\]
A projection operator, on the other hand, allows to access layers of a layered architecture configuration based on their ports.
\begin{defn}\label{def:proj}
Given a layered architecture configuration $c$ and a port $p \in \lin l \cup \lout l$ for some $l \in \layers c$, we define the layer projection
\[
	\proj p c = l.
\]
\end{defn}\noindent
By Def.~\ref{def:lac}, the layer possessing a given port is unique, so $\proj{.}{.}$ is well-defined.

\subsection{Semantics}
%%In a layered architecture configuration, the semantics of a layer may be influenced by the semantics of its dependent layers. 
%Informally speaking, the actual functioning of a layer in a layered architecture configuration usually depends on the actual valuation of the open input-ports and on the actual functioning of the layers attached to the layer.
%%Thus, in the following, we define another notion of a layers, its \emph{configuration dependent semantics}.
Now we are going to define the computational meaning of a layered architecture configuration.

In the following, for a map $f\colon X{\to}Y$, we write $f|_{Z}$ for the restriction of $f$ to the domain $X{\cap}Z$.
\begin{defn}
\label{def:composition}For a layered architecture configuration $c$, the \emph{attachment-closure} $\lout l^*$ of the output ports of a layer $l$ is
\begin{align}
	\lout l^* = \bigcap\bigl\{ P \subseteq \sel c \mid & \ \lout{l}\subseteq P\label{eq:semPortLayerOut}\\
%	& \wedge (\forall\,i\in\sel[i]c\setminus\sel[in]c\colon i\in P\Rightarrow \fconf c i \in P)\label{eq:semPortIn}\\
	& \wedge (\forall\,(i,o)\in\conf c\colon i\in P\Rightarrow o\in P)\label{eq:semPortIn}\\
	& \wedge (\forall\,o\in\sel[o]c\colon o \in P \Rightarrow \lin{\proj o c}\subseteq P)\bigr\}.\label{eq:semPortOut}
\end{align}
The \emph{configuration semantics} of a layer $l\in\layers c$ is a function $\sem cl\colon\val{\sel[\mathit{in}] c}\rightarrow\pset{\val{\lout l}}$, 
with
%\begin{align}
%\llbracket l\rrbracket_{c}\left(\mu\right) \ = \ \bigl\{ \nu|_{\lout l} \mid & \exists\,P\colon\quad \lout{l}\subseteq P \subseteq\sel{c} \ \wedge\ \nu\in\val{P}\label{eq:semLinePorts}\\
% & \ \wedge\ \mu|_{\sel[\mathit{in}]{c}\cap P} = \nu|_{\sel[\mathit{in}]{c}\cap P}\label{eq:semLineInput}\\
% & \ \wedge \ (\forall\,(i,o)\in\conf{c}\colon i\in P\Rightarrow (o\in P\wedge \nu(i)=\nu(o)))\label{eq:semLineConf}\\
% & \ \wedge \ (\forall\,r\in\layers{c}\colon \lout{r}\cap P\neq\emptyset\label{eq:semLineChooseLayer}\\
% & \ \ \quad\Rightarrow(\lin{r}{\subseteq}P\,\wedge\,\exists\,\xi\in\lfun{l}(\nu|_{\lin{r}})\colon\xi|_P=\nu|_{\lout{r}}))\bigr\}\label{eq:semLineLayerBeh}
%\end{align}
\begin{align}
\sem c l\left(\mu\right) \ = \ \bigl\{ \nu|_{\lout l} \mid & \ \nu\in\val{\lout l^*}\label{eq:semLineVal}\\
 & \ \wedge\ \mu|_{\lout l^*} = \nu|_{\sel[\mathit{in}]{c}}\label{eq:semLineOpenIn}\\
 & \ \wedge \ (\forall i\, \in \sel[i]c \cap \lout l^* \colon (\nu(i)=\nu(\fconf c i)))\label{eq:semLineIn}\\
 & \ \wedge \ (\forall o\in \sel[o]c \cap \lout l^* \colon\nonumber\\
 & \qquad \exists\,\xi\in\lfun{\proj o c}(\nu|_{\lin{\proj o c}})\colon \xi|_{\lout l^*}=\nu|_{\lout{\proj o c}})\bigr\}\label{eq:semLineOut}
\end{align}

\end{defn}
In~\eqref{eq:semLineVal}, we would not like to use all the $\sel{c}$ instead of $\lout l^*$, since, informally speaking, there might be no consistent valuation of all the ports, but there may be a consistent valuation of a subset of ports that is sufficient to define the output of the layer. Instead we use the minimal set of ports including $\lout l$ and closed under the attachment relation.

%% \begin{note}\label{note:inductive}
%% \setcounter{equation}{0}
%% The set $\lout l^*$ can also be constructed recursively. Let 
%% %\begin{align}
%% \[\begin{array}{rl}
%% f\colon \pset{\sel c} \to \pset{\sel c}\,,\quad P \mapsto\ & \lout{l}
%% \\
%% & \cup \left\lbrace o \mid \exists\,i \in P \colon \left( i,o\right) \in\conf c \right\rbrace
%% \\
%% & \cup \left\lbrace i \in \lin r \mid r \in \layers c \,\land\, P \cap \lout r \neq \emptyset \right\rbrace.
%% \end{array}\]
%% Using the fixpoint theorem of Tarski one can show that $\lout l^*=\bigcup_{i\in \mathbb{N}} f^i\left( \emptyset \right)$.
%% \end{note}
Each element of the semantics $\sem c l(\mu)$ is created by constructing a valuation $\nu$ of the ports $\lout l^*$ of the configuration that are needed for getting the value of the output ports of $l$ and projecting $\nu$ to these output ports. In fact, line \eqref{eq:semLineVal} says that $\nu$ provides a valuation of all needed ports. Line \eqref{eq:semLineOpenIn} says that the valuation of an open input-port must be taken into account if and only if we need this port. Line \eqref{eq:semLineIn} says that if we require the value of a connected input-port, then we use the value of the corresponding output-port. Line \eqref{eq:semLineOut} says that if we need a service provided by a layer, then the computation proceeds according to the layer's behavior function.

\begin{exmp}[Calculating a layer's configuration semantics]\label{ex:semantics}
Consider,
for example, the layered architecture configuration $c=(\{l_f,l_g\},A)$ in Fig.\ \ref{fig:Composition1}.

\begin{figure}%
\vspace{-.8ex}
\centering%
\begin{subfigure}[b]{0.4\textwidth}\centering%
\begin{tikzpicture}[iport/.style={circle,fill,inner sep=2pt},oport/.style={circle,draw,fill=white,inner sep=2pt}, layer/.style={draw,minimum height=20pt,minimum width=90pt}]
%\draw[dashed]  (-2.5,2) rectangle (2.5,-0.5);
iffalse Layer L1--------------------------------------------------------------- fi
\node[layer] (l1) at (0,1.5) {$l_g=(I_g,O_g,g)$};
\node [oport] at ($(l1)+(-45pt,5pt)$) {}; \node at ($(l1)+(-55pt,5pt)$) {$o_1$};
\node [iport] at ($(l1)+(-45pt,-5pt)$) {}; \node at ($(l1)+(-55pt,-5pt)$) {$i_1$};
\node [oport] (l1o1) at ($(l1)-(30pt,10pt)$) {}; \node at ($(l1o1)+(1.17ex,-1.27ex)$) {$o_1'$};
\node [iport] (l1i1) at ($(l1)+(30pt,-10pt)$) {}; \node at ($(l1i1)+(1.17ex,-1.27ex)$) {$i_1'$};
iffalse Layer L0--------------------------------------------------------------- fi
\node[layer] (l0) at (0,0) {$l_f=(I_f,O_f,f)$};
\node [iport] at ($(l0)-(45pt,5pt)$) {}; \node at ($(l0)-(+55pt,5pt)$) {$i_0$};
\node [oport] at ($(l0)+(-45pt,5pt)$) {}; \node at ($(l0)+(-55pt,5pt)$) {$o_0$};
\node [iport] (l0i1) at ($(l0)+(-30pt,10pt)$) {}; \node at ($(l0i1)+(-1.17ex,1.25ex)$) {$i_0'$};
\node [oport] (l0o1) at ($(l0)+(30pt,10pt)$) {}; \node at ($(l0o1)+(-1.4ex,1.25ex)$) {$o_0'$};
\draw[-latex] (l1o1) edge (l0i1) ; \draw[-latex]  (l0o1) edge (l1i1) ; \end{tikzpicture}
\caption{Before update.}%
\label{fig:Composition1}%
\end{subfigure}%
\qquad
\begin{subfigure}[b]{0.4\textwidth}\centering%
\begin{tikzpicture}[iport/.style={circle,fill,inner sep=2pt},oport/.style={circle,draw,fill=white,inner sep=2pt}, layer/.style={draw,minimum height=20pt,minimum width=90pt}]
%\draw[dashed]  (-2.5,2) rectangle (2.5,-0.5);
iffalse Layer L1--------------------------------------------------------------- fi
\node[layer] (l1) at (0,1.5) {$l_g=(I_g,O_g,g)$};
\node [oport] at ($(l1)+(-45pt,5pt)$) {}; \node at ($(l1)+(-55pt,5pt)$) {$o_1$};
\node [iport] at ($(l1)+(-45pt,-5pt)$) {}; \node at ($(l1)+(-55pt,-5pt)$) {$i_1$};
\node [oport] (l1o1) at ($(l1)-(30pt,10pt)$) {}; \node at ($(l1o1)+(1.17ex,-1.27ex)$) {$o_1'$};
\node [iport] (l1i1) at ($(l1)+(30pt,-10pt)$) {}; \node at ($(l1i1)+(1.17ex,-1.27ex)$) {$i_1'$};
iffalse Layer L0--------------------------------------------------------------- fi
\node[layer] (l0) at (0,0) {$l_f=(I_f,O_f,f')$};
\node [iport] at ($(l0)-(45pt,5pt)$) {}; \node at ($(l0)-(+55pt,5pt)$) {$i_0$};
\node [oport] at ($(l0)+(-45pt,5pt)$) {}; \node at ($(l0)+(-55pt,5pt)$) {$o_0$};
\node [iport] (l0i1) at ($(l0)+(-30pt,10pt)$) {}; \node at ($(l0i1)+(-1.17ex,1.25ex)$) {$i_0'$};
\node [oport] (l0o1) at ($(l0)+(30pt,10pt)$) {}; \node at ($(l0o1)+(-1.4ex,1.25ex)$) {$o_0'$};
\draw[-latex] (l1o1) edge (l0i1) ; \draw[-latex]  (l0o1) edge (l1i1) ; \end{tikzpicture}
\caption{After update.}%
\label{fig:Composition2}%
\end{subfigure}%
\caption{Layered architecture configuration consisting of two layers $l_{f}$ and $l_{g}$.% Note: this is not a basic layered architecture, since we have a cyclic syntactic dependency (see Section~\ref{sec:BasicVariant}).}\label{fig:Composition
}%
\end{figure}
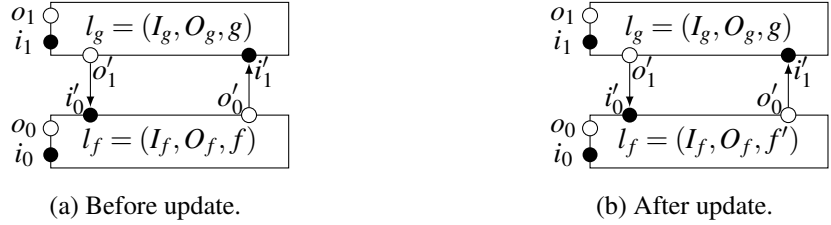%
\noindent
Here, $l_f = (\{i_0,i_0'\},\{o_0,o_0'\}, f)$ and $l_g = (\{i_1,i_1'\},\{o_1,o_1'\}, g)$
with %$\{i_{0},i_{0}',i_{1},i_{1}',\\o_{0},o_{0}',o_{1},o_{1}'\}\subseteq\mathtt{PORT}$ and
$\{i_{0},i_{0}',i_{1},i_{1}'\}\subseteq\mathcal{I}$, $\{o_{0},o_{0}',o_{1},o_{1}'\}\subseteq\mathcal{O}$, and $A=\{(i_0',o_1'),(i_1',o_0')\}$.

For the sake of this example, let's assume that $\{B,C,D,F,X,Y\}\subseteq\mathtt{SERVICE}$
and 
$\mathit{type}\left(i_0\right)=\left\{ B\right\}$, 
$\mathit{type}\left(i_0'\right) 
= \mathit{type}\left(o'_1\right) = \left\{ X, Y \right\}$,
$\mathit{type}\left(i_1\right)=\left\{ D\right\}$,
$\mathit{type}\left(i_1'\right) 
= \mathit{type}\left(o_0'\right) = \left\{ X, Y\right\}$, and
$\mathit{type}\left(o_0\right)=\mathit{type}\left(o_1\right)=\left\{ C,F\right\}$. 
%$\mathit{type}\left(o_0'\right)=\left\{ X, Y\right\}$,
%$\mathit{type}\left(o_1\right)=\left\{C,F\right\}$, and
%$\mathit{type}\left(o_1'\right)=\left\{X, Y\right\}$. 
Here, we use symbols $B,C,D,F$ for services at externally visible ports ($i_0,o_0,i_1,o_1$) and $X,Y$ for services at internal ports ($i_0',o_0',i_1',o_1'$).

The behavior functions are as follows:
\[
\begin{array}{r@{\quad}ll}
f\colon\overline{\left\{ i_{0},i_0'\right\}} \to\pset{\overline{\left\{ o_{0},o_0'\right\}}},& 
\left[i_{0},i_0'\mapsto B,X\right] &\mapsto \{\left[o_{0},o'_{0}\mapsto C,X\right]\},\\
&\left[i_{0},i_0'\mapsto B,Y\right] &\mapsto \{\left[o_{0},o'_{0}\mapsto F,X\right]\};\\
g\colon\overline{\left\{ i_{1},i_1'\right\} }\to\pset{\overline{\left\{ o_{1},o_1'\right\} }},& 
\left[i_{1},i_1'\mapsto D,X\right] &\mapsto \lbrace\left[o_1,o'_1\mapsto C,X\right]\rbrace,\\
&\left[i_1,i_1'\mapsto D,Y\right] &\mapsto \{\left[o_1,o'_1\mapsto F,Y\right]\}.
\end{array}
\]
Let us now apply Def.\ \ref{def:composition} to calculate $\sem c{l_f}(\mu)\subseteq \overline{\left\{o_0, o_0'\right\}}$
and $\sem c {l_g}(\mu)\subseteq \overline{\left\{o_1, o_1'\right\}}$ for $\mu = \left[i_{0},i_{1}\mapsto B,D\right]$.
Therefore, we first calculate all elements $\nu\in\val{\sel{c}}$. For our simple system, $\nu$ satisfies
\[
\begin{array}{l@{\quad}l}
\nu\left(i_{0}\right) = \mu(i_0) = B\,, & 
\nu\left(i_{1}\right) = \mu(i_1) = D\,, \\
\nu\left(i_0'\right)=\nu\left(o_1'\right) = X\,, &
\nu\left(i_1'\right)=\nu\left(o_0'\right) = X\,, \\
\nu\left(o_{0}\right) = C\,, &
\nu\left(o_{1}\right) = C.
\end{array}
\]
%Note that $\nu$ satisfies the constraints that should hold for each element of $\sem c{.}(\mu)$ according to Def.\ \ref{def:composition} and that $\nu$ is the only element of $\val{\sel{c}}$ which does so. 
Note that $\nu$ is the only element of $\val{\sel{c}}$ that satisfies the constraints that should hold for each element of $\sem c{.}(\mu)$ according to Def.\ \ref{def:composition}.
Thus,
\[\begin{array}{rl}
\sem c {l_f}(\mu) &=\lbrace\nu|_{\{o_0,o'_0\}}\rbrace =\lbrace\left[o_0,o'_0\mapsto C,X\right]\rbrace \ \text{and}\\
\sem c {l_g}(\mu) &=\lbrace\nu|_{\{o_1,o'_1\}}\rbrace =\lbrace\left[o_1,o'_1\mapsto C,X\right]\rbrace.
\end{array}\]
\qed
\end{exmp}

%Note that Def.~\ref{def:composition} is not constructive: it does not tell us how to calculate the configuration semantics for a layer, but only how to check whether a particular valuation is indeed valid. Later on, we provide a direct, recursive way of determining a layer's configuration semantics.

\subsection{Semantic Change}

%In a layered architecture configuration, the semantics of a single layer may change over time. 
A key concept in developing a piece of software is changing the semantics of a layer.
We model such a change of the semantics of a layer through an update function.
\begin{defn}
For a layer $l$ and a map $f\colon \val{\lin l} \to \pset{\val{\lout l}}$, a \emph{semantic update} $\lupdate lf$ is the layer $\left(\lin l,\lout l,f\right)$.
\end{defn}\noindent
Note that a semantic update is indeed a layer according to Def.~\ref{def:layer}.

The notion of semantic update easily generalizes to sets of layers $L\subseteq \mathcal{L}$:
\[\update Llf = \left(L \setminus l\right) \cup \{\lupdate lf\}.\defc\]
Finally, it also generalizes to layered architecture configurations:
\[\update clf = \left(\update {\layers c}lf,\conf c\right).\defc\]

\begin{exmp}[A semantic update for a layered architecture configuration]
Consider, for example, the layered architecture configuration $c = (L,A)$
depicted in Fig.~\ref{fig:Composition1} and described in Ex.~\ref{ex:semantics}.

If we change the behavior of layer $l_f$ to
\[
\begin{array}{r@{\quad}l}
f'\colon\overline{\left\{ i_{0},i_0'\right\}} \to\pset{\overline{\left\{ o_{0},o_0'\right\}}},& %~\mathrm{with}\
\left[i_{0},i_0'\mapsto B,X\right] \mapsto \{\left[o_{0},o'_{0}\mapsto C,X\right]\},\\
&\left[i_{0},i_0'\mapsto B,Y\right] \mapsto \{\left[o_{0},o'_{0}\mapsto F,Y\right]\},
\end{array}
\]
we get a new layered architecture configuration $\update c{l_f}{f'}$ where layer $l_f$ has changed to $(I_f,O_f,f')$ (see Fig~\ref{fig:Composition2}).
Applying Def.\ \ref{def:composition} to calculate $\sem {\update c{l_f}{f'}}{l_f}(\mu)$ $\subseteq$ $\overline{\left\{o_0,o_0'\right\} }$ and $\sem {\update c{l_f}{f'}}{l_g}(\mu)\subseteq \overline{\left\{o_1,o_1'\right\} }$ for $\mu = \left[i_{0},i_{1}\mapsto B,D\right]$, produces, in addition to $\nu$, a new valuation $\nu'$, which satisfies
\[
\begin{array}{l@{\quad}l}
\nu'\left(i_{0}\right) = \mu(i_0) = B\,, & 
\nu'\left(i_{1}\right) = \mu(i_1) = D\,, \\
\nu'\left(i_0'\right)=\nu'\left(o_1'\right) = Y\,, &
\nu'\left(i_1'\right)=\nu'\left(o_0'\right) = Y\,, \\
\nu'\left(o_{0}\right) = F\,, &
\nu'\left(o_{1}\right) = F\,.
\end{array}
\]
Note that $\nu'$ satisfies the constraints that should hold for each element of $\sem {\update c{l_f}{f'}}{.}(\mu)$ according to Def.\ \ref{def:composition} and that $\nu,\nu'$ are now the only elements of $\val{\sel{c}}$ which do so. 
Thus,
\[\begin{array}{rl}
\sem {\update c{l_f}{f'}}{l_f}(\mu) &=\lbrace \nu|_{\{o_0,o'_0\}}, \nu'|_{\{o_0,o'_0\}}\rbrace =\lbrace\left[o_0,o'_0\mapsto C,X\right],\left[o_0,o'_0\mapsto F,Y\right]\rbrace\ \text{and}\\
\sem {\update c{l_f}{f'}}{l_g}(\mu) &=\lbrace\nu|_{\{o_1,o'_1\}},\nu'|_{\{o_1,o'_1\}}\rbrace =\lbrace\left[o_1,o'_1\mapsto C,X\right],\left[o_1,o'_1\mapsto F,Y\right]\rbrace\,.
\end{array}\]\qed
\end{exmp}
In the above example as well as in general, a semantic update of a layered architecture configuration changes neither the input/output-ports nor the attachment, thus producing a layered architecture configuration again:
\begin{prop}
For a layered architecture configuration $c$, layer $l\in\layers c$,
and a map $f\colon \val{\lin l} \rightarrow \pset{\val{\lout l}}$, 
the layered architecture configuration update $\update clf$ is a layered architecture configuration.
\end{prop}
Thus, all properties and notation introduced so far for layered architecture configurations are also valid for layered architecture configuration updates.

\subsection{Syntactic Dependency}

In a layered architecture configuration, the attachment relation induces a dependency relation between layers.
We say that a layer $l'$ \emph{syntactically} depends on another layer $l$, if an input port of $l'$ is connected to an output port of $l$.
\begin{defn}\label{def:syndep}
	Syntactic dependency for a layered architecture configuration $c$ is a relation $\syndep c\ \subseteq\ \layers c \times \layers c$ defined by
	\[
	\syndep[l][l']c\quad\stackrel{\text{def}}{\Leftrightarrow}\quad\exists\ o\in\lout{l},\ i\in\lin{l'}\colon\ \ o = \fconf c i.
	\]
\end{defn}
\begin{exmp}[Syntactic dependency]
  In the layered architecture configuration depicted in Fig.\
  \ref{fig:layered-architecture-configuration},	we have
  $\syndep[l_i][l_{i{+}1}]c$ for $i\in\{0,\ldots,n{-}1\}$ and no other syntactic
  dependencies.\qed
\end{exmp}
For a layered architecture configuration $c$, we denote by $\syndept c$ % $\ \subseteq\ \layers c \times \layers c$
the \emph{transitive} closure of %relation
$\syndep c$ and by $\syndeprt c$ %$\ \subseteq\ \layers c\times\layers c$
the \emph{reflexive-transitive} closure of $\syndep c$. Moreover,
we denote by $\syndep[][\_]c:\,\layers c\rightarrow\pset{\layers c}$,
defined via
\[
\syndep[][m]c=\{l\in\layers c\mid\syndep[l][m]c\}\qquad \text{for}\ m\in\layers c\,,\defc
\]
all layers $l$ that a given layer $m$ syntactically depends on ($\syndept[][\_]c$, $\syndeprt[][\_]c$ for [reflexive-] transitive dependency, respectively).

\begin{lem}\label{lem:ports}
  For a layered architecture configuration $c$, and layer $l\in \layers c$, the attachment closure $\lout l^*$ contains only ports of layers on which layer $l$ reflexively-transitively syntactically depends on. Formally, $\forall p\in \lout l^*\colon \syndeprt[{\proj p c}][l]c$.
\end{lem}
\begin{proof}
  Let $f\colon \pset{\sel c} \to \pset{\sel c}$,\\
  $P \mapsto\ 
  \lout{l}
  \cup \left\lbrace o \mid \exists\,i{\in}P \colon \left( i,o\right) \in\conf c \right\rbrace
  \cup \bigcup\left\lbrace \lin r \mid r \in \layers c \,\land\, P \cap \lout r \neq \emptyset \right\rbrace$.\\
  Using the fixed point theorem of Tarski one can show that $\lout l^*=\bigcup_{n\in \NNZ} f^n(\emptyset)$.
  %Since $\lout l^*$ can be constructed inductively, the proof is by induction according to Note~\ref{note:inductive}.
  Fix a layered architecture configuration $c$ and one of its layers $l \in \layers c$. We show that $\forall n\in\NNZ\ \forall\,p\in f^n(\emptyset) \colon {\syndeprt[{\proj p c}][l]c}$ by induction on $i$.
  \begin{compactitem}
  \item[``$\forall p\in f^0(\emptyset) \colon {\syndeprt[{\proj p c}][l]c}$'':] Since $f^0(\emptyset)=\emptyset$, the statement is vacuously true.
  \item[``$\forall p\in f^n(\emptyset) \colon {\syndeprt[{\proj p c}][l]c}$ implies $\forall p\in f^{n+1}(\emptyset) \colon {\syndeprt[{\proj p c}][l]c}$'':]
  Fix $p \in f^{n+1}(\emptyset)$. %According to Note~\ref{note:inductive}, $p\in\lout{l}$ or $p \in \left\lbrace o \mid \exists i \in f^n\left( \emptyset \right) \colon \left( i,o\right) \in\conf c \right\rbrace$ or $p \in \left\lbrace i \in \lin r \mid r \in \layers c \land f^n\left( \emptyset \right) \cap \lout r \neq \emptyset \right\rbrace$.
%According to Note~\ref{note:inductive}, at
At 
least of the following cases is true.
  
  Case $p\in\lout{l}$: By Def.~\ref{def:proj}, $\proj p c=l$ and by reflexivity, $\syndeprt[l][l]c$. Thus, $\syndeprt[{\proj p c}][l]c$.
  
  Case $p \in \left\lbrace o \mid \exists i \in f^n(\emptyset) \colon \left( i,o\right) \in\conf c \right\rbrace$: Then there is an $i \in f^n(\emptyset)$ such that $\left(i,p\right)\in \conf c$. By Def.~\ref{def:proj} and~\ref{def:lac} we have $i\in \lin{\proj i c}$ and $p \in \lout{\proj p c}$. From $\left(i,p\right)\in \conf c$ we obtain $\syndep[{\proj p c}][{\proj i c}]c$ by Def.~\ref{def:syndep}. Since $i \in f^n(\emptyset)$, we have $\syndeprt[{\proj i c}][l]c$ by induction hypothesis. By transitivity, $\syndeprt[{\proj p c}][l]c$.
  
  Case $p \in  \bigcup\left\lbrace \lin r \mid r \in \layers c \land f^n(\emptyset) \cap \lout r \neq \emptyset \right\rbrace$: Then there is an $r\in \layers c$ such that $f^n(\emptyset) \cap \lout r \neq \emptyset$ and $p \in \lin r$. Since $f^n(\emptyset) \cap \lout r \neq \emptyset$, we have $\syndeprt[r][l]c$ by induction hypothesis. Since $p \in \lin r$, we have $\proj p c=r$ by Def.~\ref{def:proj}. Thus, we conclude $\syndeprt[{\proj p c}][l]c$.\vspace{-12pt}
  \end{compactitem}
\end{proof}

\noindent
Note that the syntactic dependency relation is not transitive in general: just because a layer $L_1$ depends on another layer $L_2$ which depends on a third layer $L_3$, this does not necessarily mean that layer $L_1$ depends on layer $L_3$.

\subsection{Semantic Dependency}
Besides the syntactic dependency relation between layers of a layered architecture configuration we also have a semantic dependency relation between those layers. A layer $l'$ semantically depends on a layer $l$ if updating $l$ may influence the configuration semantics of $l'$.
\begin{defn}\label{def:semdep}
Semantic dependency for a layered architecture configuration $c$ is a relation $\semdep c\ \subseteq\ \layers c\times\layers c$ defined by
\[
\begin{array}{r@{\ \ }c@{\ \ }l@{\qquad}l}
  \semdep[l][l]c & \stackrel{\mathrm{def}}{\Leftrightarrow} & \exists f\in\left(\val{\lin l}\to\pset{\val{\lout l}}\right)\colon \sem c{l}\neq\sem{\update clf}{\lupdate lf} & \text{for all}\ l\ \text{in}\ \layers c\ \text{and}\\
  \semdep[l][l']c & \stackrel{\mathrm{def}}{\Leftrightarrow} & \exists f\in\left(\val{\lin l}\to\pset{\val{\lout l}}\right)\colon \sem c{l'}\neq\sem{\update clf}{l'} & \text{for all}\ l\neq l'\ \text{in}\ \layers c\,.
\end{array}
\]
\end{defn}\noindent
%
%In the following we provide a simple example for the notion of semantic dependency.
Now we provide simple examples of semantic dependency and independence.
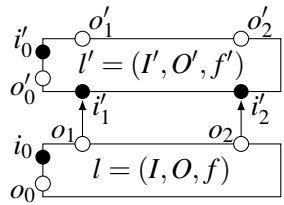
\begin{figure}[htbp]%
\centering%
\begin{tikzpicture}[iport/.style={circle,fill,inner sep=2pt},oport/.style={circle,draw,fill=white,inner sep=2pt},
layer/.style={draw,minimum height=20pt,minimum width=90pt}]
%\draw[dashed]  (-4,10.5) rectangle (2.5,0.5);
\node[layer] (l1) at (-0.5,3.4) {$l'=(I',O',f')$};

\node [iport] (l1i0) at ($(l1)+(-45pt,5pt)$) {};
\node at ($(l1i0)+(-1.5ex,.8ex)$) {$i_0'$};

\node [oport] (l1o0) at ($(l1)-(45pt,5pt)$) {};
\node at ($(l1o0)+(-1.5ex,-.8ex)$) {$o_0'$};

\node [oport] (l1o1) at ($(l1)+(-30pt,10pt)$) {};
\node at ($(l1o1)+(1.5ex,1.3ex)$) {$o_1'$};

\node [oport] (l1o2) at ($(l1)+(30pt,10pt)$) {};
\node at ($(l1o2)+(1.5ex,1.3ex)$) {$o_2'$};

\node [iport] (l1i1) at ($(l1)-(30pt,10pt)$) {};
\node at ($(l1i1)+(1.5ex,-1.3ex)$) {$i_1'$};

\node [iport] (l1i2) at ($(l1)+(30pt,-10pt)$) {};
\node at ($(l1i2)+(1.5ex,-1.3ex)$) {$i_2'$};

iffalse Layer L0--------------------------------------------------------------- fi

\node[layer] (l0) at (-0.5,2) {$l=(I,O,f)$};

\node [iport] (l0i0) at ($(l0)+(-45pt,5pt)$) {};
\node at ($(l0i0)+(-1.5ex,.8ex)$) {$i_0$};

\node [oport] (l0o0) at ($(l0)-(45pt,5pt)$) {};
\node at ($(l0o0)+(-1.5ex,-.8ex)$) {$o_0$};

\node [oport] (l0o1) at ($(l0)+(-30pt,10pt)$) {};
\node at ($(l0o1)+(-1.4ex,.8ex)$) {$o_1$};

\node [oport] (l0o2) at ($(l0)+(30pt,10pt)$) {};
\node at ($(l0o2)+(-1.5ex,.8ex)$) {$o_2$};

\draw[-latex] (l0o1) edge (l1i1) ;
\draw[-latex] (l0o2) edge (l1i2);
\end{tikzpicture}%
\caption{A simple layered architecture configuration with $2$ layers.}%
\label{fig:simple-layered-architecture-configuration}%
\end{figure}%

\begin{exmp}[Semantic dependency]
\label{ex:lowerdoesinfluenceupper}
As an example, consider the simple layered architecture configuration depicted in Fig.~\ref{fig:simple-layered-architecture-configuration} where changing the behavior of layer $l$ does indeed influence the configuration semantics of layer $l'$.

In order to see this, we first need to formally define the behavior functions $f \in \left(\val I\rightarrow\pset{\val O}\right)$ and $f' \in \left(\val {I'}\rightarrow\pset{\val {O'}}\right)$. Let us assume that $\{A,B,C,D,E,F,X,Y\}$ $\subseteq$ $\mathtt{SERVICE}$
and $\mathit{type}(i_0)=\{A\}$, $\mathit{type}(o_0)=\{B\}$, $\mathit{type}(o_1)=\mathit{type}(i_1')=\{X\}$, $\mathit{type}(o_2)=\mathit{type}(i_2')=\{Y,Z\}$, $\mathit{type}(i_0')=\{C\}$, $\mathit{type}(o_0')=\{D\}$, $\mathit{type}(o_1')=\{E\}$, and $\mathit{type}(o_2')=\{F,G\}$. Here, we use symbols $A,B,C,D,E,F$ for services occurring at externally visible ports ($i_0,o_0,i_0',o_0', o_1',o_2'$) and $X,Y$ for services occurring at internal ports ($o_1,o_2,i_1',i_2'$).

Let $f:\overline{\{i_0\} }\to\pset{\overline{\{o_0,o_1,o_2\}}}$
and $f':\overline{\{i_0',i_1',i_2'\}} \to \pset{\overline{\{o_0',o_1',o_2'\}}}$ be defined by%
\[\begin{array}{rl}
f\left(\left[i_{0}\mapsto A\right]\right) & =\left\{[o_{0},o_{1},o_{2}\mapsto B,X,Y]\right\}\,,\\
f'\left(\left[i_0',i_{1}',i_{2}'\mapsto C,X,Y\right]\right) & =\left\{[o_{0}',o_{1}',o_{2}'\mapsto D,E,F]\right\}\,,\\
f'\left(\left[i_0',i_{1}',i_{2}'\mapsto C,X,Z\right]\right) & =\left\{[o_{0}',o_{1}',o_{2}'\mapsto D,E,G]\right\}\,.
\end{array}\]
Now we calculate $\sem cl:\overline{\{i_0,i'_0\}}\to\pset{\overline{\{ o_0,o_1,o_2\}}}$
and $\sem c{l'}:\overline{\{ i_0, i_0'\} }\to\pset{\overline{\{ o_0',o_1',o_2'\}}}$ 
by Def.\ \ref{def:composition}: 
\[\begin{array}{rl}
\sem cl\left([i_0, i'_0\mapsto A, C]\right)&=\lbrace\left[o_0, o_1,o_2\mapsto B,X,Y \right]\rbrace\,,\\
\sem c{l'}\left([i_0,i_0'\mapsto A, C]\right)&=\lbrace\left[o_0', o_1',o_2'\mapsto D,E,F \right]\rbrace\,.
\end{array}\]
If we now replace $f$ by $g:\overline{\left\{i_0\right\}}\to\pset{\overline{\left\{o_0,o_1,o_2\right\}}}$, defined as
\[g\left([i_0\mapsto A]\right) = \lbrace\left[o_{0},o_{1},o_{2}\mapsto B,X,Z\right]\rbrace\,,\]
we can see that $\semdep[l][l']c$, because calculating
$\sem{\update clg}{l}:\val{\{i_0,i'_0\}}\to\pset{\overline{\{o_0,o_1,o_2\}}}$ and $\sem{\update clg}{l'}:\overline{\{i_0,i_0'\}}\to\pset{\val{\{o_0',o_1',o_2'\}}}$ by Def.~\ref{def:composition} results in
\[\begin{array}{rl}
\sem{\update clg}{l}\left([i_0,i'_0\mapsto A, C]\right)&=\lbrace\left[o_0, o_1,o_2\mapsto B,X,Z \right]\rbrace\,,\\
\sem{\update clg}{l'}\left([i_0, i_0'\mapsto A, C]\right)&=\lbrace\left[o_0', o_1',o_2'\mapsto D,E,G \right]\rbrace\,.
\end{array}\]
and we see that $\sem{\update clg}{l'}\neq \sem{}{l'}$.\qed
\end{exmp}

\begin{exmp}[Semantic independence]
In the simple layered architecture configuration in Fig.~\ref{fig:simple-layered-architecture-configuration} changing the behavior of layer $l'$ does not influence the configuration semantics of layer $l$.

Let us assume behavior functions $f:\val I\rightarrow\pset{\val O}$ and $f':\val {I'}\rightarrow\pset{\val {O'}}$ of Ex.~\ref{ex:lowerdoesinfluenceupper}. Then we can see that there is no behavior function $g:\val {I'}\rightarrow\pset{\val {O'}}$ such that 
$\sem{\update{c}{l'}{g}}{l}\neq \sem{}{l}$. This is the case, because the semantics of $l$ does not depend on any inputs from $l'$. Thus, we have $\notsemdep[l'][l]c$.\qed
\end{exmp}

\subsection{Relating Syntactic and Semantic Dependencies}
Having a formal model of layered architecture configurations allows us to %further 
analyze the relationship between syntactic and semantic dependencies.

An interesting property is that if layers are syntactically dependent, this does not necessarily mean that they are also semantically dependent.
\begin{exmp}[Syntactic dependency does not necessarily imply semantic dependency]
Consider a single layer with just one input and just one output port that are 
typed by the empty set of services and attached to each other. According to 
Def.~\ref{def:syndep}, the layer depends on itself syntactically. However, it 
is not possible to change the layers configuration semantics at all, since the 
layer's behavior function is the only map from the (empty) set of valuations 
of the input port to the (empty) set of valuations of the output port. Thus, 
according to Def.~\ref{def:semdep}, the layer does not depend on itself 
semantically. In general, if $\mathtt{SERVICE}=\emptyset$, any configuration 
with a nonempty attachment will have a pair of layers with this property.\qed
\end{exmp}

However, under certain circumstances, syntactic dependency does indeed imply semantic dependency.
\begin{defn}
  A layered architecture configuration $c$ is \emph{usable} iff there is at
  least one valuation of open input-ports such that the configuration semantics
  of every layer produces at least one output valuation on this input. Formally:\\
  $c\ \text{usable}\quad \stackrel{\mathrm{def}}{\Longleftrightarrow}\quad \exists\, \mu \in \val{\sel[\mathit{in}] c} \ \forall\, l \in \layers c \colon \sem cl\left( \mu\right) \not=\emptyset$.
\end{defn}
%Informally speaking, a layered architecture configuration is usable iff
%all layers do something under some input.
%
\begin{thm}
	For a usable layered architecture configuration $c$ the reflexive-transitive closure of syntactic dependency implies semantic dependency. Formally: $c$ usable $\Rightarrow$ $\syndeprt c \subseteq \semdep c$.
\end{thm}
\begin{proof}
  Let $c$ be usable and $\syndeprt[l][l']c$. So there is some $\mu \in\val{\sel[\mathit{in}] c}$ such that $\sem c{l'}(\mu)\neq\emptyset$. Let $g:\val{\lin l}\to\pset{\val{\lout l}}$, $i\mapsto\emptyset$. If $l=l'$, then $\sem{\update{c}{l'}{g}}{\lupdate{l'}{g}}(\mu) =\emptyset$. If $l\neq l'$, we inductively follow that all the layers $r\neq l$ such that $\lout l\cap\lout r^*\neq\emptyset$ satisfy $\sem{\update clg}{r}(\mu)=\emptyset$. In particular, $\sem{\update clg}{l'}\left(\mu\right)=\emptyset$.
\end{proof}

Vice versa, if layers are semantically dependent, they are not necessarily (directly) syntactically dependent.
\begin{exmp}[Semantic dependency does not necessarily imply syntactic dependency]%
Consider a single layer with one output port that is typed by two services and no other ports. According to Def.~\ref{def:semdep}, it depends on itself semantically. However, according to Def.~\ref{def:syndep}, it does not depend on itself syntactically. Indeed, it does not have any syntactic dependency at all.

A less trivial example is demonstrated in Fig.~\ref{fig:threeLayers}, where $\mathtt{SERVICE}=\{A,B\}$, all ports are typed by $\mathtt{SERVICE}$, $\lfun l=\lambda \nu{\in}\val{\emptyset}.\,\{[o\mapsto A]\}$, $\lfun {l'}=\lambda \nu\in\val{\{i'\}}.\,\{[o'\mapsto\nu(i')]\}$, and $\lfun {l''}=\lambda \nu\in\val{\{i''\}}.\,\{[o''\mapsto\nu(i'')]\}$. We have $\semdep[l][l'']c$, but $\nsyndep[l][l'']c$.
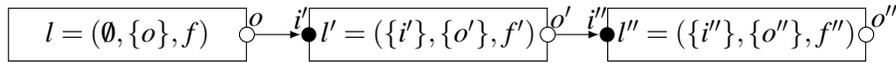
\begin{figure}[htbp]%
\centering%
\begin{tikzpicture}[iport/.style={circle,fill,inner sep=2pt},oport/.style={circle,draw,fill=white,inner sep=2pt},layer/.style={draw,minimum height=20pt,minimum width=90pt}]
\node[layer] (l2) at (10.1,1) {$l''=(\{i''\},\{o''\},f'')$};
\node[oport] (o2) at ($(l2.east)$){};
\node at ($(o2)+(7pt,7pt)$) {$o''$};
\node[iport] (l2i) at (l2.west){};
\node at ($(l2i)+(-.13,.22)$) {$i''$};
\node[layer] (l1) at (6,1) {$l'=(\{i'\},\{o'\},f')$};
\node[oport] (l1o) at ($(l1.east)$){};
\node at ($(l1o)+(.18,.23)$) {$o'$};
\node[iport] (l1i) at ($(l1.west)$){};
\node at ($(l1i)+(-.1,.22)$) {$i'$};
\node[layer] (l0) at (2,1) {$l=(\emptyset,\{o\},f)$};
\node[oport] (l0o) at ($(l0.east)$){};
\node at ($(l0o)+(.12,.2)$) {$o$};
\draw[-latex] (l0o) -- (l1i);
\draw[-latex] (l1o) -- (l2i);
\end{tikzpicture}%
\caption{A three-layered configuration with $\layers c=\{l,l',l''\}$ and $\conf c$ as shown.}%
\label{fig:threeLayers}%
\vspace{-.8ex}
\end{figure}\qed%
\end{exmp}
\noindent
As we see, changing the behavior of a single layer may impact not only the configuration semantics of directly depending layers, but also of layers which transitively depend on the modified layer.

This property of layered architecture configurations implies that a test after a change of a layers behavior should include tests of the behavior of all semantically dependent layers. As we will see in a moment, there is a bound on how many layers one should test. 
\begin{thm}
Semantic dependency implies the reflexive-transitive closure of syntactic dependency. Formally: $\semdep{c}\subseteq\syndeprt{c}$.
\end{thm}
\begin{proof}Fix a layered architecture configuration $c$, its layers $l,l'\in \layers c$ such that $\nsyndeprt[l][l']c$; we will show $\notsemdep[l][l']c$.\\
Notice that $l\neq l'$. Fix arbitrary $f\colon\val{\lin l} \to \pset{\val{\lout l}}$ and $\mu\in\val{\sel[\mathit{in}] c}$. We are going to show that $\sem c{l'}(\mu)=\sem{\update clf}{l'}(\mu)$.
\begin{compactitem}
\item[``$\subseteq$'':] Let $\kappa \in\sem c{l'}(\mu)$. By \eqref{eq:semLineVal} of Def.~\ref{def:composition} there is some $\nu\in\val{\lout l^*}$ such that $\nu|_{\lout{l'}}=\kappa$ and  \eqref{eq:semLineOpenIn}, \eqref{eq:semLineIn}, \eqref{eq:semLineOut} hold for $\kappa$ and $l'$. Since $l'$ does not reflexively-transitively syntactically depend on $l$, by Lemma~\ref{lem:ports} no ports of $l$ are in $\lout {l'}^*$ and we readily conclude that \eqref{eq:semLineVal}, \eqref{eq:semLineOpenIn}, \eqref{eq:semLineIn}, \eqref{eq:semLineOut} still hold for $\kappa$ and $l'$ if $c$ is replaced by $\update clf$. Thus $\kappa \in \sem{\update clf}{l'}(\mu)$.
\item[``$\supseteq$'':]Analogously.\vspace{-15pt}
\end{compactitem}
\end{proof}\noindent
Informally speaking, this property allows us now to restrict testing after a modification to only those layers which (reflexively-)transitively depend on the modified layer.
\begin{corl}For usable layered architecture configurations, the semantic dependency and the reflexive-transitive closure of the syntactic dependency are the same.
\end{corl}

%% file: Conclusion.tex
\section{Conclusion}\label{sec:conclusion}
With this work we provided an abstract model for the layered architecture
style. Our model is based on the notion of services and ports which
can supply services. A layer consists of input and output ports and
is modeled as a function from input-port valuations to output-port
valuations. A layered architecture configuration consists then of
some layer instances and an attachment describing the connections between
layers' input and output ports.

We have given a formal definition of syntactic and semantic dependency between
layers. Though syntactic and semantic dependencies do not necessarily imply
one another, we have shown that the semantic dependency implies
the reflexive-transitive closure of the syntactic dependency, and the reverse
also holds for usable configurations.
%We identify a base variant
%which requires the syntactic dependency relation to be acyclic and
%show that it allows for a recursive definition of layer semantics.
%Finally we provide a proof method for basic layered architectures
%based on well-founded induction.

Having developed a formal model of layered architectures, the model can now be used for a rigorous analysis of the style. Thus, future work arises in two main areas:
\begin{inparaenum}[(i)]
	\item First of all, different variants of the style should be identified and defined through constraints over our model. For example, a ``basic'' variant of the style would impose a \emph{well-foundedness} constraint on the attachment relation and a ``strict'' variant would further constrain the attachment relation to be \emph{antitransitive}.
	\item Then, for each variant, a set of properties should be formulated and proved from the constraints. For example, in the ``basic'' variant, we may want to provide conditions that ensure that the configuration is usable. Moreover, the configuration semantics of lower level layers may be strictly independent of the behavior of upper level layers and under certain circumstances, the configuration semantics of upper level layers may also be independent of the behavior of lower level layers. In the ``strict'' version, changing a layers behavior may have even less impact on the configuration semantics of other layers within the architecture configuration.
\end{inparaenum}

Our work aims to contribute to a rigorous theory of architectural styles to provide a better understanding of architectural styles and the formal relationships between architectural design decisions and quality attributes. Thus, two further directions for future work arise:
\begin{inparaenum}[(i)]
	\item The approach used in this article should be applied to other architectural styles as well.
	\item Then, a general theory of architectural styles should be developed to investigate relationships between the different styles.
\end{inparaenum}

\section{Acknowledgments}
\vspace{-1ex}
%\paragraph{Acknowledgements}
This work was partially funded by the German Federal Ministry of Education and Research (BMBF), grants ``Software Campus project RE4SoS, 01IS12057'', and ``ARAMiS project, 01IS11035''.

We would like to thank Manfred Broy, Wolfgang Boehm, Maximilian Irlbeck, Maximilian Junker, Andreas Vogelsang, Vasileios Koutsoumpas, Veronika Bauer, and Daniel M\'endez Fern\'andez for their comments and helpful suggestions.